%% file: denittis-polo2.tex
\documentclass[graybox]{svmult}

\usepackage{enumerate}
\usepackage{type1cm}        % activate if the above 3 fonts are
  \usepackage{amsmath}                        % not available on your system
\usepackage{slashed}
\usepackage{makeidx}         % allows index generation
\usepackage{graphicx}        % standard LaTeX graphics tool
\usepackage{amssymb} 
\usepackage{mathrsfs}% when including figure files

\newcommand{ \ii}{\,\mathrm{i}\,}
\newcommand{\expo}[1]{\,\mathrm{e}^{#1}\,} 

\newcommand{\B}{\mathcal{B}}
\newcommand{\PE}{\widetilde{P}_m^\varepsilon}
\newcommand{\LL}{\mathfrak{L}}
\newcommand{\Ran}{{\rm Rand}}
\newcommand{\G}{\mathcal{G}}
\newcommand{\h}{\mathfrak{H}}
\newcommand{\R}{\mathbb{R}}
\newcommand{\N}{\mathbb{N}}
\newcommand{\C}{\mathbb{C}} 

\newcommand{\Z}{\mathbb{Z}}
\newcommand{\T}{\mathbb{T}}
\newcommand{\pr}{\mathbb{P}}
\newcommand{\Aff}{\mathtt{AFF}}
\newcommand{\A}{\mathcal{A}}

\newcommand{\Tp}{\mathscr{T}_\pr}

\newcommand{\ef}{\epsilon_F}
\makeindex             % used for the subject index
\begin{document}

\title*{Topological polarization in disordered systems}
% Use \titlerunning{Short Title} for an abbreviated version of
% your contribution title if the original one is too long
\author{Giuseppe De Nittis and Danilo Polo Ojito}
% Use \authorrunning{Short Title} for an abbreviated version of
% your contribution title if the original one is too long
\institute{Giuseppe De Nittis \at Pontificia Universidad Catolica de Chile, Santiago de Chile-Chile, \email{gidenittis@mat.uc.cl}
\and Danilo Polo Ojito \at Pontificia Universidad Catolica de Chile, Santiago de Chile-Chile \email{djpolo@mat.uc.cl}}
%
% Use the package "url.sty" to avoid
% problems with special characters
% used in your e-mail or web address
%
\maketitle

\abstract*{Deformations in piezoelectric materials lead to conduction effects, which are due to two contributions: the relative displacements of the ionic cores,  and the so-called orbital polarization.
This work is devoted to the rigorous derivation of the celebrated
King-Smith and Vanderbilt formula for orbital polarization in a generalized setting that includes continuous random systems among others.
}

\abstract{Deformations in piezoelectric materials lead to conduction effects, which are due to two contributions: the relative displacements of the ionic cores,  and the so-called orbital polarization.
This work is devoted to the rigorous derivation of the celebrated
King-Smith and Vanderbilt formula for orbital polarization in a generalized setting that includes continuous random systems among others.}

%--------%

\section{Introduction}

In nature, there are materials in which a \emph{macroscopic polarization} at the edges of the sample appears when subjected to mechanical strains, \emph{i.e.}, to the accumulation of charge whenever the materials are deformed. This phenomenon is known as \emph{piezoelectric effect}, and its microscopically description was only understood in the last 50 years. In the $70'$s, Martin \cite{Mar} noticed that the previous approach in terms of dipole momenta for the macroscopic polarization was incomplete and unsatisfactory. This fact is due to that the  total polarization comes from two contributions: the \emph{relative displacements} of the ionic cores in a unit cell (whose computation is straightforward), and electrical conduction which is called \emph{orbital polarization}.
Resta \cite{Res1} and King-Smith and Vanderbilt \cite{KSV}  shifted the attention to the orbital polarization and derived a formula using linear response theory, which allows calculating the polarization in terms of the \emph{Berry connection}. Namely, the change in polarization $\Delta \mathscr{P}$ accumulated during a (periodic) deformation in the interval $[0,T]$ is given by
\begin{equation}\label{KSV}
    \Delta \mathscr{P}\;:=\;\frac{1}{(2\pi)^d}\sum_{m=0}^M\int_{\mathbb{B}}{\rm d}k\big(\mathscr{A}_m(k,T)-\mathscr{A}_m(k,0)\big).
\end{equation}
Here $\mathbb{B}\simeq\mathbb{T}^d$ denotes the first Brillouin zone, $d$ is the space dimension, $\mathscr{A}_m(k,t)$ is the Berry connection for the $m$-th Bloch band at time $t$, and the sum runs over all the occupied $M$ Bloch bands.
Panati, Sparber, and Teufel \cite{Pan} generalized equation \eqref{KSV} for continuous and periodic systems, showing that in the adiabatic limit of slow deformations the macroscopic piezoelectric current is determined by the geometry of the Bloch bundle. Using an adaption of Nenciu’s super-adiabatic theory \cite{Nen} to $C^*$-dynamical systems, Schulz-Baldes and Teufel \cite{Shulz} established  formula \eqref{KSV} for discrete random systems. They obtained that in the adiabatic limit it holds true  that
\begin{equation}\label{KSV2}
    \Delta\mathscr{P}_k\;=\;\ii\int_0^T {\rm d}t\mathscr{T}\big( P(t) [\, \partial_t P(t), \,\nabla_k P(t) \,]\big) +\mathcal{O}(\varepsilon^N ),
\end{equation}
where $\mathscr{T}$ denotes the trace per unit volume, $P(t)=\chi_{(-\infty,\ef)}\big(H(t)\big)$ is the spectral projection onto all states below the \emph{Fermi energy} $\ef$, $H(t)$ is the instantaneous
Hamiltonian of the system at time $t$,  $N\in\N$ is related to the regularity of the map $t\mapsto H(t)$ and $k=1,\ldots,d$  indicates the direction of the polarization in the physical space.
It is important to point out that the works \cite{Pan} and \cite{Shulz} it is also explored the topological nature of orbital polarization. They proved that $\Delta \mathscr{P}$ is quantized up to a small error (in the \emph{adiabatic parameter} $\varepsilon$) whenever the slow deformation is periodic. The latter fact  is in agreement with the observation of Thouless \cite{Thou} in a more restricted context. In \cite{Deni2} one of the authors and Lein  carried out a topological study of the orbital polarization in discrete graphene-like systems, where they showed that the polarization depends only on the class of homotopy paths in the gapped parameter space. Therefore, a necessary condition for the existence of piezoelectric effects is that the fundamental group of the gapped parameter space is non-trivial.

\medskip

In this work, we focus on deriving rigorously  equation \eqref{KSV2} for continuous and disordered systems of independent electrons in the regime of an adiabatic periodic deformation of the background potential. The main strategy is to use the mathematical framework introduced in \cite{De Nittis}, along with tools from (super)adiabatic theory
\cite{Nen,Teu}, for the derivation of the formula for $\Delta \mathscr{P}$ in a wide range of \emph{covariant} (random) systems,  which in principle are defined over a topological group $\G$ that can be chosen equal to $\R^d$ (continuous case) or  $\Z^d$ (discrete case) in concrete applications.
Our main result, Theorem \ref{polarizacio},   establishes the expression for the orbital polarization in this generalized setting, along  with its main topological consequences  when the deformation is periodic. 

\medskip
\noindent
{\bf Organization of the paper.} Section \ref{sect:rand} is devoted to the construction of the semi-finite von Neumann algebra of observables and its trace per unit volume. In Section \ref{sec:hip}, we briefly review all the necessary mathematical notions and we state the main hypotheses needed for the derivation of equation \eqref{KSV2}.
In Section \ref{KS-for} we present the main results. We start this section with a notion of differentiability for affiliated self-adjoint operators to the observable algebra, and after that, we prove
an equivalence for the current expectation value (Theorem \ref{Expresion}). We will use the later facts to derive the {\em King-Smith and Vanderbilt formula}. We finish this section with the topological quantization of the polarization for periodic deformations.
Section \ref{apll} provides the physical models where our results apply. We will present in detail the case of continuous disordered systems and we will build the Landau Hamiltonian which fulfills all the required hypotheses.
In order to maintain the clarity in the proof of the Theorem \ref{polarizacio}, in Appendix \ref{append} we have included the technical proofs needed for the construction of the superadiabatic projections.

%--------%

\section{Description of the physical models}\label{sect:rand} 
The background material contained in this section is based on \cite[Chapter 4]{De Nittis} where the relevant references are also provided.

\medskip

Let $\mathfrak{h}$ be a (separable) Hilbert space  and  $\mathcal{B(\mathfrak{h})}$  the set of linear bounded operators on $\mathfrak{h}$. The physical relevant \emph{observables} (like the Hamiltonians) will be modeled by  strongly continuous\footnote{In the sense of the resolvent.} families $\big(H_\omega)_{\omega\in \Omega}$ of (self-adjoint) operators affiliated to a von Neumann algebra $\A$. Here  $\Omega$ denotes a compact\footnote{We will assume that $\Omega$ is also metrizable, and in turn separable. This assumption implies that $L^2(\Omega)$ is a separable Hilbert space.}   space which describes the possible configurations
of the \emph{interacting potential} between particles and medium (e.g. random interaction). In order to construct a von Neumann algebra $\A$ which contains homogeneous models\footnote{In the sense of Bellissard \cite{Bel}.}  we assume that there is an \emph{ergodic topological dynamical system} $(\Omega, \G,\tau, \pr)$ consisting of a locally compact\footnote{In the interesting examples $\G$ is also separable and metrizable (e.g. $\G=\R^d,\Z^d,\T^d$) and this implies that $L^2(\G)$ is a separable Hilbert space.} abelian group $\G$ (with a given Haar measure $\mu_\G$), a probability space $(\Omega,\mathcal{F}, \pr)$, where $\mathcal{F}$ is the Borel $\sigma$-algebra and $\pr$ is a probability measure, and a representation $\tau:\G\rightarrow {\rm Homeo}(\Omega)$. These structures are related by the following assumptions:
\begin{enumerate}[(i)]
    \item The group action $\G\times\Omega\ni(g,\omega)\mapsto \tau_g(\omega)\in\Omega$ is jointly continuous;
    \smallskip
\item $\pr$ is a $\tau$-invariant ergodic measure, namely $\pr(\tau_g(B))=\pr(B)$ for all $B\in \mathcal{F}$, and if $\tau_g(B)=B$ for all $g\in \G$ then $\pr(B)=1$ or $\pr(B)=0$.
\end{enumerate}

\medskip

In the next we will consider the Hilbert space
$\mathfrak{h}=L^2(\G)\otimes \C^N,$
where $N$ depends on the spin-type degrees of freedom (e.g. the isospin) and we will introduce the direct integral \cite[Part II, Chapters 1-5]{Dix1}
$$
{\h}\;:=\;\int_\Omega^{\oplus}{\rm d}\;\pr(\omega) \;\mathfrak{h}_\omega\;\simeq\; L^2(\Omega,\mathfrak{h})\;,
$$
with the assumption that $\mathfrak{h}_\omega=\mathfrak{h}$ for (almost) all $\omega\in\Omega.$ A random operator is a bounded-operator valued map $\Omega \ni \omega\mapsto A_\omega\in \B(\mathfrak{h})$ such that
the map $\Omega \ni\omega\mapsto \langle \phi,A_\omega\psi\rangle_\mathfrak{h}$ is measurable for all $\phi,\psi\in \mathfrak{h}$,
and ${\rm ess}-\sup\|A_\omega\|_{\B(\h)}<\infty$. We will denote the set of random operators by $\Ran({\h})\subset \B({\h})$. Furthermore, any random operator ${A}:=\{A_\omega\}_{\omega\in\Omega}$ fulfills 
$$\|{A}\|_{\B({\h})}\;=\;{\rm ess}-\sup_{\omega\in\Omega}\|A_\omega\|_{\B(\mathfrak{h})}\;.$$

\medskip

Let $\Theta\colon\G\times\G\to \mathbb{U}(1)$ be a \emph{twisting group $2$-cocycle} \cite[Definition 4.1.2]{De Nittis},  and
for every $g\in\G$ consider the operator ${U}_g\in{\B({\h})}$ defined by
 $$ \big({U}_g\ {\psi}\big)_{\tau_g(\omega)}(h)\;:=\;\Theta(g,hg^{-1})\;\psi_\omega(hg^{-1})\;,\qquad \forall\; h\in\G$$
 where ${\psi}:=\{\psi_\omega\}_{\omega\in\Omega}$ is any element of ${\h}$, and on the left-hand side the symbol $(\cdot)_{\tau_g(\omega)}$  means that the value of the vector ${U}_g{\psi}$ on the fiber of $\h$ at $\tau_g(\omega)$. It is evident from the definition that ${U}_g$ doesn't respect the fiber structure of the direct integral ${\h}$.
   One can check that the mapping $ \G\ni g\mapsto  {U}_g\in {\B( {\h})}$ forms a projective unitary representation of $\G$.
\begin{definition}
The von Neumann algebra of observables is the set
$$\A\;\equiv\;\A(\Omega,\pr,\G,\Theta)\:=\; {\rm Span}_\G\{ {U}_g\}'\;\cap\; \Ran( {\h})
$$
where ${\rm Span}_\G\{ {U}_g\}$ denotes the linear space generated by the $ {U}_g$ and the symbol $'$ denotes the commutant.
\end{definition}
For sake of notational simplicity, we write $\A$ instead of $\A(\Omega,\pr,\G,\Theta)$. 
Said differently  $\A$ consists of those random operators ${A}$ which are covariant with respect to the projective unitary representation of $\G$ provided by the  ${U}_g$, i.e.
$$
U_{g,\tau_g(\omega)}\;A_\omega\;U_{g,\tau_g(\omega)}^{-1}\;=\;A_{\tau_g(\omega)}\;,\qquad \forall\; g\in\G\;,\;\; \forall\; \omega\in\Omega
$$
where $U_{g,\tau_g(\omega)}$ denotes the action of ${U}_g$ from the fiber at $\omega$ into the fiber at $\tau_g(\omega)$.

\medskip

It is known that $\A$ is a semi-finite von Neumman algebra, hence $\A$ admits a faithful normal semi-finite (f.n.s.) trace \cite[Part I, Chapter 6, Proposition 9]{Dix1}. On the domain of definition, such a trace can be constructed following the procedure described in \cite[Proposition 2.1.6 and Theorem 2.2.2]{Len}, i.e.
$$
\mathscr{T}_\pr({A})\;:=\;\int_\Omega {\rm d}\pr(\omega) \,{\rm Tr}_{\mathfrak{h}}(M_\lambda A_\omega M_\lambda)\;,\qquad {A}\in\A^+\;,
$$
where $\lambda\in L^\infty(\G)\cap L^2(\G)$ is any positive  function 
of unitary norm $\|\lambda\|_{L^2}=1$,
and ${M}_\lambda$ is the  operator  which acts on $\mathfrak{h}$ as the multiplication by the diagonal matrix $\lambda\otimes{\bf 1}_N$.
It turns out that  $\Tp$ coincides with the \emph{trace per unit volume}, namely
$$\mathscr{T}_\pr({A})\;=\;\lim_{n\rightarrow\infty}\frac{1}{|\Lambda_n|}\,{\rm Tr}_{\h}(P_{\Lambda_n}A_\omega P_{\Lambda_n})\;,\qquad \pr-{\rm a. e.}
$$
 where $P_{\Lambda_m}$ is the multiplication operator by the characteristic function of the compact set $\Lambda_n\subset\G$, $|\Lambda_n|$ is its volume, and $\{\Lambda_n\}_{n\in\N}$ forms a \emph{F{\o}lner exhausting sequence} for $\G$.

%-----------------------------%
\section{Main hypotheses for a linear response theory}
\label{sec:hip}
In this section, we will briefly review all the necessary mathematical notions and we will state the main hypotheses  needed for a rigorous derivation of the  linear response theory as formulated in  \cite[Section 2]{De Nittis}.  These notions and hypotheses will be  used in the following sections of this work.

\medskip

Let $\Aff(\A)$ be the set of  \emph{closed} and \emph{densely defined} operators affiliated with $\A$ \cite[Section 3.1.2]{De Nittis}, and  $\LL^p(\A)$ the  \emph{$L^p$-spaces}  (or  $p$-Schatten classes) associate to the semi-finite von Neumman algebra $\A$ with its f.n.s. trace $\mathscr{T}_\pr$   \cite{Dix1,Nel,Ter, Yao} or \cite[Section 3.2]{De Nittis}. The non-commutative H\"older inequalities allow defining the commutators
$$
[A,B]_{(r)}\;:=\;AB-BA\;\in\;\LL^r(\A)\;,\qquad A\in \LL^p(\A)\;,\;\;B\in \LL^q(\A)\;,
$$
with $r^{-1}=p^{-1}+q^{-1}$.

\medskip

\noindent
\textbf{Hypothesis 1 (unperturbed dynamics).} Let $H\in \Aff(\A)$ be a (possibly unbounded) self-adjoint operator (or \emph{Hamiltonian})  which prescribes the \emph{unperturbed dynamics} of the system.  The affiliation of $H$ to $\A$ implies that the unperturbed dynamics
$$\alpha_t(A)\;:=\;\expo{-\ii tH }A\expo{\ii tH},\qquad\;t\in \R,\;A\in \A,$$
generated by $H$ is a strongly continuous one-parameter group of isometries on each Banach space $\LL^p(\A)$. The generator $\mathscr{L}_H^{(p)}$ of $\alpha_t$ on  $\LL^p(\A)$ has a core $\mathcal{D}_{H,p}$ where it acts as a generalized commutator
\cite[Proposition 5.1.3]{De Nittis}, i.e.
$$\mathscr{L}_H^{(p)}(A)\;=\;-\ii\big( HA-(HA^*)^*\big)\;=:\;-\ii[H,A]_{\dag}\;,\qquad A\in\mathcal{D}_{H,p}\;.
$$
we will refer to $\mathscr{L}_H^{(p)}$ as the \emph{$p$-Liouvillian}
of $H$. An (initial) \emph{equilibrium configuration} for $H$ is any positive element $\rho\in \A^+$ such that $\alpha_t(\rho)=\rho$ for every $t\in\R$. It will be called \emph{equilibrium state} if in addition $\mathscr{T}_\pr(\rho)=c<+\infty$ (and up to a multiplicative factor one can always impose the normalization condition $c=1$). 
For instance $\rho=f(H)$, with $f\in L^\infty(\R)$ and positive, is an equilibrium configuration.

%\medskip
%\noindent
% \textbf{Hypothesis 2 (gap condition).} There is a spectral subset $\sigma_*\subset\sigma\big(H\big)\subset\R$ such that  $${\rm dist}\big(\,\sigma(H), \sigma (H)\setminus \sigma_*\,\big)\;=\;\Delta\;>\;0\;.
% $$
% We will denote by $P_*$  the spectral projection of $H$ for the relevant \emph{gapped} part of the spectrum $\sigma_*$. 

\medskip
\noindent
 \textbf{Hypothesis 2. (spatial derivation)} Let   $\{ X_1,X_2,...,X_d\}$ be a set of (possibly unbounded) self-adjoint operators which are  \emph{$\Tp$-compatible} in the sense that for all $k=1,2,...,d$ and for all $s\in \R$ they satisfy
\begin{enumerate}[(i)]
    \item $\expo{\ii sX_k}A\expo{-\ii sX_k}\in \A$ for all $A\in \A$;
        \smallskip
    \item $\Tp(\expo{\ii sX_k}A\expo{-\ii sX_k})=\Tp(A)$ for all $A$ in the domain of $\Tp$;
        \smallskip
    \item $\expo{\ii sX_j}\expo{\ii sX_k}=\expo{\ii sX_k}\expo{\ii sX_j}$ for all $j,k=1,2,...,d$.
\end{enumerate}
 This assumption allows to introduce the spatial derivations on $\LL^p(\A)$ as generators of an $\R$-flow, i.e.
$$\nabla_k(A)\;:=\;\lim_{s\rightarrow 0}\frac{\expo{\ii sX_k}A\expo{-\ii sX_k}-A}{s}.$$
The $\nabla_k$ are densely defined closed operators on each $\LL^p(\A)$ with a common core where they act as commutators, i.e. $\nabla_k(A)=\ii[X_k,A]$ \cite[Section 3.4.1]{De Nittis}. The domain of the associated gradient $\nabla:=(\nabla_{1},...,\nabla_{d})$ is the (non-commutative) \emph{Sobolev space} \cite[Section 3.4.2]{De Nittis}
  $$\mathfrak{M}^{1,p}(\A)\;:=\;\{ A\in \LL^p(\A)\,|\, \nabla_k(A)\in \LL^p(\A), \,k=1,2,...,d\}.$$

\medskip
\noindent \textbf{Hypothesis 3 (current operator).} The self-adjoint Hamiltonian $H\in \Aff(\A)$ with dense domain $\mathcal{D}(H)$ and the set of $\Tp$-compatible generators $\{X_1,X_2,...,X_d\}$ with common localizing domain $\mathcal{D}_c\subset {\h}$  \cite[Remark 3.4.7]{De Nittis} meet the following assumptions:
\begin{enumerate}[(i)]
    \item The joint core $\mathcal{D}_c(H):=\mathcal{D}_c\cap \mathcal{D}(H)$ is a densely defined core for $H$, and $X_k[\mathcal{D}_c(H)]\subset \mathcal{D}_c(H)$ for all $k=1,\ldots,d$;
    \smallskip
    \item $H[\mathcal{D}_c(H)]\subset \mathcal{D}_c$ and the formal commutators
    \begin{equation}
        J_k\;:=\;-\ii(X_kH-HX_k)\;,\qquad k=1,\ldots,d
       % (-1)^{|\kappa|}\;\nabla_1^{\kappa_1}\circ\ldots\circ\nabla_d^{\kappa_{d}}(H)\;=\;-\ii[X_k,H]\;,\qquad \kappa \in \mathbb{N}_0^d\;,
    \end{equation}
    %where $\N_0:=\{0\}\cup\N$,
    are essentially self-adjoint on $\mathcal{D}_c(H)$, and therefore uniquely extend to self-adjoint operators denoted (with abuse of notation) by $J_k=\nabla_k(H)$. 
%        \smallskip
%         \item There exists a $N\in\N_0$, called the \emph{order} of $H$ with respect the $X_k$'s, such that $J_\kappa=0$ for all multi-indices 	$\kappa$ such that $|\kappa|>0$;
          \smallskip
    \item All the $J_k$ are infinitesimally $H$-bounded, i.e., for any $\delta>0$ there are constants $a>0$ and $\delta>b>0$ such that
    $$\|J_k\varphi\|_{\h}\leq a\|\varphi\|_{\h}+b\|H\varphi\|_\h,\,\,\,\,\,\,\varphi\in\mathcal{D}_c(H)$$
    for all $k=1,\ldots,d$.
      \smallskip
    \item $J_k\in \Aff(\A)$ for every $k=1,\ldots,d$.
\end{enumerate}
The vector-valued operator
    $$
    J\;:=\;\nabla(H)\;=\;(\nabla_1(H),\ldots,\nabla_d(H))
    $$
    will be called \emph{current operator}.

\medskip
\noindent
\textbf{Hypothesis 4 (perturbed dynamics).} Let $\R\supseteq I\ni t\mapsto H(t)\in \Aff(\A)$ be a path such that: 
\begin{enumerate}[(i)] 
 \item $H(0)=H$ and $\mathcal{D}(H(t))=\mathcal{D}(H)$ for every $t\in\R$;
 \medskip
    \item For every $t\in I$ the operator $H(t)$ meets the properties of Hypotheses 3 and therefore there exists the time-dependent current $J(t)\;=\;\nabla(H(t))$;
     \medskip
    \item There exists a unique
strongly jointly continuous map $I^2\ni(s,t)\mapsto U(s,t)\in \A$, called    
     \emph{unitary propagator}, which leaves invariant the domain $\mathcal{D}(H)$ and solves the differential equation
     $$
     \ii\partial_t \psi_s(t)\;=\;H(t)\psi_s(t)\;,\qquad \psi_s(s)=\psi_0\in \mathcal{D}(H)
     $$
     in the sense that $\psi_s(t)=U(t,s)\psi_0$.
\end{enumerate}
The unitary propagator verifies the properties $U(t,t)={\bf 1}$ 
and $U(t,s)U(s,r)=U(t,r)$ for every $t,s,r\in I$.
Suitable conditions for the existence of the unitary propagator are given in  \cite[Theorem 5.2.4]{De Nittis}. Since $U(t,s)\in \A$,   it can be used to define  dynamics on $\A$ and $\LL^p(\A)$  by
\begin{equation}
    \eta_{(t,s)}(A)\;:=\;U(t,s)AU(s,t)\;,\qquad (t,s)\in\R^2\;,\quad A\in \A\;\; \text{or}\;\;\LL^p(\A).
\end{equation}
These are isometries jointly strongly continuous in $t$ and $s$ on $ \A$ and in each $\LL^p(\A)$. Moreover, it turns out that the map $I\ni t\mapsto \eta_{(t,s)}(A)\in\LL^p(\A)$ is differentiable for every fixed $s$, and
$$
 \ii\partial_t \eta_{(t,s)}(A)\;=\;[H(t),\eta_{(t,s)}(A)]_{\dag}
$$
whenever $HA$ and $HA^*$ are in  $\LL^p(\A)$ \cite[Proposition 5.2.6]{De Nittis}.

\medskip
\noindent
 \textbf{Hypothesis 5 (gap condition).} Let $\sigma_*(t)\subset\sigma(H(t))$ be a subset of spectrum of $H(t)$ such that there exist continuous function $f_{\pm}:I\rightarrow \R$ defining  intervals  $G(t)=[f_-(t),f_+(t)]$ so that $\sigma_*(t)\subset G(t)$ and $$g\;:=\;\inf_{t\in I}{\rm dist}\big(G(t), \sigma (H(t))\setminus \sigma_*(t)\big)$$
 is strictly positive. 
We will denote by $P_*(t):=\chi_{\sigma_*(t)}(H(t))$  the spectral projection of $H(t)$ on the gapped spectral patch $\sigma_*(t).$

 \medskip
\noindent
 \textbf{Hypothesis 6 (regularity of the equilibrium state).} Let $\rho$ be an  equilibrium state for $H$. We assume that $\rho$ is $p$-regular, i.e.
\begin{enumerate}[(i)]
    \item $\rho\in \A^+\cap \mathfrak{M}^{1,1}(\A)\cap\mathfrak{M}^{1,p}(\A)$;
    \vspace{1mm}
    \item  $H(t)\rho(t)\in \mathfrak{M}^{1,1}(\A)\cap\mathfrak{M}^{1,p}(\A)$ for all $t\in I$.
\end{enumerate}
The state $\rho$ can be evolved also by the perturbed dynamics $\eta_{(t,s)}$ through the prescription
\begin{equation}
     \rho(t)\;:=\;\eta_{(t,0)}(\rho)\;=\;U(t,0)\rho U(0,t)\;,\qquad t\in \R.
\end{equation}
Since $\rho(t)^*=\rho(t)$ for every $t\in I$ it follows that the generalized commutator $[H(t),\rho(t)]_{\dag}$ is well defined and 
from \cite[Theorem 5.2.6]{De Nittis} one gets that  $\rho(t)$ is the unique solution of 
\begin{equation}\label{eq_eq_t}
   \left\{ \begin{array}{lcc}\;
             
 \ii\partial_t\rho(t)\;=\;[H(t),\rho(t)]_{\dag}\;,
\\\\
           \;\;\rho(0)=\rho\;,
             \end{array}
   \right.
\end{equation}
where the derivative is taken in $\LL^1(\A)$ or $\LL^p(\A)$.

%---------------%

 \section{The King-Smith and Vanderbilt formula for the orbital polarization}\label{KS-for}
In this section we present the main results of this paper, i.e., the derivation of the \emph{King-Smith and Vanderbilt formula} for the orbital polarization. \medskip

\noindent
Let us start by saying that a self-adjoint map  $\R\supseteq I\ni t\mapsto H(t)\in \Aff(\A)$ is $N$-differentiable in the uniform sense (in the interval $I$) if the map 
$$
I\;\ni\; t\longmapsto\;\big(\ii{\bf 1}-H(t)\big)^{-1}\;\in\;\A
$$ 
is $N$-differentiable with respect to the norm topology of $\A$.
We will denote with $C^N(I,\A)\subset\A$ the space of $\A$-valued maps which are  $N$-differentiable.

\begin{remark}\label{Rem 1}
Notice that if the map $I\ni t\mapsto H(t)\in\Aff(\A)$ is $N$-differentiable in the uniform sense, then it is also true that 
$$\big(z{\bf 1}-H(\cdot)\big)^{-1}\;\in\; C^N(I,\A
)$$
for each  $z\in\C$ which lies in the resolvent set of $H(t)$, for any $t\in I$. Indeed, one has that
$$\big(z{\bf 1}-H(t)\big)^{-1}-\big(\ii{\bf 1}-H(t)\big)^{-1}\;=\;(\ii-z)\big(z{\bf 1}-H(t)\big)^{-1}\big(\ii{\bf 1}-H(t)\big)^{-1}.$$
Thus,
$$\big(z{\bf 1}-H(t)\big)^{-1}\;=\;F(z,t)\big(\ii{\bf 1}-H(t)\big)^{-1}$$
where $$F(z,\cdot)\;:=\;\left(\,{\bf 1}-(\ii-z)\big(\ii{\bf 1}-H(\cdot)\big)^{-1}\,\right)^{-1}\;\in\; C^N(I,\A)\;.$$
Therefore
$$\partial_t^n\big(z{\bf 1}-H(t)\big)^{-1}\;=\;\partial_t^nF(z,t)\big(\ii{\bf 1}-H(t)\big)^{-1}+F(z,t)\partial^n_t\big(\ii{\bf 1}-H(t)\big)^{-1}$$
for $0<n\leq N$ in consequence of the fact that $(\ii-z)^{-1}$ lies in the resolvent of $\big(\ii{\bf 1}-H(t)\big)^{-1}$ for every $t\in I$, and of the identity
$$\partial_t^n F(z,t)=-(\ii-z)F(z,t)\partial_t^n\big(\ii{\bf 1}-H(t)\big)^{-1}F(z,t)\;,\hspace{1cm}0<n\leq N\;.$$
\end{remark}
Our first result is a generalization of \cite[Proposition 4]{Shulz}.

\begin{theorem}\label{Expresion}
 Let $\R\supseteq I\ni t\mapsto H(t)\in \Aff (\A)$ be a  path of self-adjoint operators  which meets Hypothesis 1, 2, 3, 4. Let $P\in\A$ be an orthogonal projection which satisfies Hypothesis 6 with $p=1,2$.
Let $P(t):=\eta_{(t,0)}(P)$ and $J_k(t)$ the $k$-th component of the current operator $J(t)=\nabla(H(t))\in \Aff (\A)$.
Then, the current expectation value can be rewritten as
 \begin{equation}\label{expre}
      \Tp\big(J_k(t)P(t)\big)\;=\;\ii\Tp\big(P(t)[\partial_tP(t),\nabla_k (P(t))]_{(1)}\big)
 \end{equation}
 for every $k=1,\ldots,d$.
\end{theorem}
\begin{proof}
From the hypothesis we have that $H(t)P(t)\in \LL^1(\A)$ and $J_k(t)\in \Aff(\A)$ for all $t\in I$. Then,
using \cite[Lemma 3.3.7]{De Nittis} one obtains that
$$J_k(t)P(t)\;=\;\big(J_k(t)(H(t)-z{\bf 1})^{-1}\big)\big((H(t)-z{\bf 1})P(t)\big)\in \LL^1(\A),$$
where $z$ (which can depend on $t$) lies in the resolvent set of $H(t)$.
Therefore, the left-hand side of the expressions \eqref{expre} is well defined. From the hypothesis, it also follows that 
$\nabla_k (P(t))\in \LL^2(\A)$ and $H(t)P(t)\in \LL^2(\A)$ which implies 
$$
\partial_tP(t)\;=\;-\ii\big(H(t)P(t)-(H(t)P(t))^*\big)\;\in\; \LL^2(\A)\;.
$$
Therefore, the commutator 
$[\partial_tP(t),\nabla_k (P(t))]_{(1)}$ is a well defined element in  $\LL^1(\A)$.
For sake of notational simplicity, we suppress the $t$ dependencies in the following computation. 
 From  \cite[Lemma 3.2.14]{De Nittis}, one gets
 \begin{equation*}
     \begin{split}
         \ii\Tp\big(P[\partial_tP,\nabla_k (P)]_{(1)}\big)&\;=\; \lim_{n\rightarrow \infty}\ii\Tp\big(P[\partial_tP P_n,\nabla_k (P)]_{(1)}\big)\\
         &=\;\lim_{n\rightarrow \infty}\ii\Tp\big(P[\partial_tP P_n,\nabla_k (P)]_{(1)}P\big)
     \end{split}
 \end{equation*}
where $P_n(t):=\chi_{[-n,n]}\big(H(t)\big)$  is the spectral projection of $H(t)$ on  $[-n,n]$. Moreover, one has that
 \begin{equation*}
     \begin{split}
        \ii \partial_tP P_n\;&=\;[H,P]_{\dag}P_n\;=\;HPP_n-(HP)^*P_n\\
        &=\;HPP_n-PHP_n\;=\;(HP-PH)P_n,
     \end{split}
 \end{equation*}
 since $(HP)^*=PH$ when projected on  $P_n$. Thus, beginning from the right-hand side of \eqref{expre} and using the properties of the trace one finds
\begin{equation*}
     \begin{split}
  &\ii\Tp\big(P[\partial_tP,\nabla_k (P)]_{(1)}\big)\\
  &=\lim_{n\rightarrow \infty}\Tp\big(P[(HP-PH)P_n,\nabla_k (P)]_{(1)}P\big)\\
  &=\lim_{n\rightarrow \infty}\Tp\big(P(HP-PH)P_n\nabla_k( P)P-P\nabla_k (P)(HP-PH)P_nP\big)\\
  &=\lim_{n\rightarrow \infty}\Tp\big(-PHP_n\nabla_k (P)P-P\nabla_k (P)HPP_nP\big)\\
      \end{split}
 \end{equation*}
 where in the last equality we have used the identity $P\nabla_k (P)P=0$ (which follows from $\nabla_k (P)=\nabla_k (P)^2$) to remove the term $PHPP_n\nabla_k (P) P$ which goes to $0$ when $n\rightarrow \infty$, and the term  $P\nabla_k (P) PHP_nP$. 
Since  $PHP_n\in \A$ and using the ``integration by part'' between $\Tp$ and $\nabla_k$ one gets
  \begin{equation*}
     \begin{split}
  &\ii\Tp\big(P[\partial_tP,\nabla_k( P)]_{(1)}\big)\\
      &=\lim_{n\rightarrow \infty}\Tp\big(\nabla_k(PHP_n)P\big)\;-\;\Tp\big(P\nabla_k (P)HPP_n\big)\\
      &=\;\lim_{n\rightarrow \infty}\Tp\big(\nabla_k (P)HP_nP+P\nabla_k (HP_n)P\big)-\Tp\big(P\nabla_k( P)HPP_n\big)\\
      &=\;\lim_{n\rightarrow \infty}\Tp\big(P\nabla_k (HP_n)P\big)+\lim_{n\rightarrow \infty}\Tp\big(\nabla_k( P)HP_nP-P\nabla_k( P)HPP_n\big)\\
      &=\;\lim_{n\rightarrow \infty}\Tp\big(\nabla_k (HP_n)P\big)+\lim_{n\rightarrow \infty}\Tp\big(\nabla_k( P)HP_nP-\nabla_k( P)HPP_nP\big)\\
      &=\;\Tp\big(\nabla_k( H)P\big)+0\\
      &=\;\Tp\big(J_kP\big).   \end{split}
 \end{equation*}
 This concludes the proof.
  \end{proof}
  
  \medskip
  
 Let $\rho_0:=\chi_{(-\infty,\ef]}(H)$ be the spectral projection of the Hamiltonian $H=H(0)$ with Fermi level $\ef\in\R$ in a gap of the spectrum of $H$. It is clear that $\rho_0$ is an initial equilibrium state for $H$. Let us assume that $\rho_0$ meets the regularity condition of Hypothesis 6 and let $\rho(t)$ be the solution of the equation \eqref{eq_eq_t}. The variation of the \emph{polarization} $\Delta \mathscr{P}_k$    between time $t=0$ and $t=T$
 due to the current $J_k$ in the state $\rho_0$
  is  by definition
 \begin{equation}\label{Polarizacion}
     \Delta\mathscr{P}_k\;: =\;
   \int_0^T{\rm d}t\;\Tp\big( J_k(t)\rho(t)\big)\;,\qquad k=1,...,d\;.
 \end{equation} 
 By using Theorem \ref{Expresion}
one can rewrite the quantity \eqref{Polarizacion} as follows
\begin{equation}\label{29}
    \Delta \mathscr{P}_k\;:=\;\ii \int_0^Tdt\,\Tp\big(\rho(t)[\partial_t\rho(t),\nabla_k(\rho(t))]_{(1)}\big)\;,\qquad k=1,2,\ldots,d\;.
\end{equation}
 It is important to point out that  equation \eqref{29} is not very useful in general, since it requires the  knowledge  of $\rho(t)$, which is not a function of $H(t)$. Thereby, we will use tools from \emph{adiabatic perturbation theory} adapted from \cite{Shulz}, in order to express the polarization in terms of the spectral projections of $H(t)$. For that, let us consider the Liouville equation
\begin{equation}\label{solution}\varepsilon\partial_t\rho(t)=-\ii[\,H(t),\,\rho(t)\,]_{\dag},\end{equation}
where $\varepsilon>0$ is a small adiabatic parameter. With these ingredients, we present now the main Theorem of this work, which is based on some technical results described in Appendix \ref{append}.
\begin{theorem}\label{polarizacio}
 Assume that the map $[0,T]\ni t\mapsto H(t)\in \Aff (\A)$ is $N+2$-differentiable in the uniform sense  and meets Hypothesis 1, 2, 3, 4, 5 and 6. If $\partial_t^n \big(\ii{\bf 1}-H(t)\big)^{-1}|_{t=0}=\partial_t^n \big(\ii{\bf 1}-H(t)\big)^{-1}|_{t=T}=0$ for all $0<n\leq N,$ then
\begin{equation}\label{eq_probo}
  \Delta\mathscr{P}_k\; =\;\ii\int_0^T {\rm d}t\;\Tp\big( P_* (t) [ \partial_t P_* (t),\, \nabla_k (P_* (t)) ]_{(1)}\big) +\mathcal{O}(\varepsilon^N )
\end{equation}
where $P_*(t) $ is the instantaneous spectral projection of $H(t)$
on the gapped spectral patch $\sigma_*(t)$.
\end{theorem}
\begin{proof}
From Theorem \ref{Teo 3}, there are projections $P_N^\varepsilon$ such that
$$\|P_N^\varepsilon(t)-\rho(t)\|+\|\nabla_k\big(P_N^\varepsilon(t)-\rho(t)\big)\|+\|\partial_t\big(P_N^\varepsilon(t)-\rho(t)\big)\big\|=\mathcal{O}(\varepsilon^N)\;.$$
Since $\rho$ is a $p$-regular initial equilibrium state, then  for some $\varepsilon>0$ small enough $\nabla_k H(t) P_N^\varepsilon(t)\in \LL^1(\A)$.  
Furthermore, the equation \eqref{29}, the norm bound property of the trace, and Corollary \ref{Coro 3} yield
$$\Delta \mathscr{P}_k\;=\;\ii\int_0^T {\rm d}t\;\Tp\big( P_N^\varepsilon (t) [\, \partial_t P_N^\varepsilon (t),\, \nabla_k (P_N^\varepsilon(t))]_{(1)}\big) +\mathcal{O}(\varepsilon^N )\;. $$
Now let us show that the above integral is independent of $\epsilon $. Indeed, since the first $N$ derivates of $t\rightarrow \big(\ii{\bf 1}-H(t)\big)^{-1}$ vanish at the endpoints then by Theorem \ref{Teo 3} one has that  $P_N^\varepsilon(0)=P_*(0)$ and $P_N^\varepsilon(T)=P_*(T)$. As a consequence of the dominated convergence theorem \cite[Corollary 5.8]{Bartle},  and following the same algebraic steps used in the proof of \cite[Theorem 1]{Shulz}, one gets 
\begin{equation*}
    \begin{split}
       & \partial_\epsilon\int_0^T {\rm d}t\; \Tp\big( P_N^\varepsilon [\partial_t{P}_N^\varepsilon, \nabla_k (P_N^\varepsilon) ]_{(1)}\big)\\
       &=\;\int_0^T {\rm d}t\;\Tp\big( P_N^\varepsilon [ \partial_\varepsilon\partial_t{P}_N^\varepsilon, \nabla_k (P_N^\varepsilon) ]_{(1)}+P_N^\varepsilon[\partial_t{P}_N^\varepsilon,\partial_\varepsilon \nabla_k( P_N^\varepsilon)]_{(1)}\big)\\
        &=\Tp\big(P_N^\varepsilon[\,\partial_\varepsilon P_N^\varepsilon,\nabla_k( P_N^\varepsilon)]_{(1)}\big)\Big|_0^T-\int_0^T {\rm d}t\;\Tp(P_N^\varepsilon[\partial_\varepsilon P_N^\varepsilon, \nabla_k( \partial_t{P}_N^\varepsilon)]_{(1)})\\
        &\;\;\; +\int_0^T {\rm d}t\;\Tp \big(P_N^\varepsilon[\partial_t{P}_N^\varepsilon, \nabla_k (\partial_\varepsilon P_N^\varepsilon)]_{(1)}\big)
        =0\;.
    \end{split}
\end{equation*}
 In both equalities above it was used that $\Tp(\partial_\varepsilon P_N^\varepsilon\partial_t{P}_N^\varepsilon\nabla_k P_N^\varepsilon)=0$, and the differentiability of the map $\varepsilon \mapsto P_N^{\varepsilon}$, which follows from Theorem \ref{Teo 3}, implies existence and equality of the mixed derivatives. Now, making $\varepsilon\rightarrow 0$  one obtains $P_N^{\varepsilon}(t)\to P_\ast(t)$ and in turn
$$  
\Delta\mathscr{P}_k\; =\;\ii\int_0^T {\rm d}t\;\Tp\big( P_* (t) [ \partial_t P_* (t),\, \nabla_k (P_* (t))]_{(1)}\big) +\mathcal{O}(\varepsilon^N )\;.
$$
This concludes the proof.
\end{proof}

\medskip

It is important to notice that the leading order term of \eqref{eq_probo} 
is invariant under diffeotopies. Consider a diffeotopy $F$ between the projection-valued paths $P_0$ and $P_1$, i.e., a smooth function $F\colon [0,T]\times [0,1]\rightarrow \mathfrak{M}^{1,1}(\A)$ such that $F(t,0)=P_0(t)$ and $F(t,1)=P_1(T)$ for all $t\in[0,T]$. 
By replacing $\varepsilon$ with the   diffeotopy parameter $s\in[0,1]$  in the proof of Theorem \ref{polarizacio}, one obtains immediately 
the equality
$$
  \Delta\mathscr{P}_k^0[P_0]\;=\;   \Delta\mathscr{P}_k^0[P_1]\;,\qquad k=1,2,\ldots,d
$$
where
$$
 \Delta\mathscr{P}_k^0[P_j]\;:=\;\int_0^T {\rm d}t\;\Tp\big( P_j (t) [\, \partial_t P_j (t),\, \nabla_k (P_j (t)) \,]_{(1)}\big)\;,\qquad j=0,1
$$
stands for the  leading order term of $\Delta\mathscr{P}_k$ with respect to the path $P_j$.

\medskip

The last important step consists in proving that  the   leading order term $\Delta\mathscr{P}_k^0$ is topologically quantized. This can be shown following the argument of  \cite{Shulz}.
\begin{corollary}
Under the assumptions of the Theorem \ref{polarizacio}, if  the deformation is cyclic, that is $H(0)=H(T)$, it holds true that 
$$
\Delta\mathscr{P}_k^0\;=\;2\pi\;{\rm Ch}({P}_*)
$$
where
$$
{\rm Ch}({P}_*)\;:=\;\frac{1}{2\pi \ii}\int_0^T{\rm d}t\;{\Tp}\big({P_*}[\ii\partial_t{P_*},\,\nabla_k({P_*})]_{(1)}\big)\;\in\;\Z
$$
is the Chern number of the differentiable map $[0,T]\ni t\mapsto P_*(t)\in \mathfrak{M}^{1,1}(\A)$. 
\end{corollary}
\begin{proof}
If the deformation is cyclic, we can consider  $P_*$ as a projection-valued map in the $C^*$-algebra $C(\mathbb{S}^1)\otimes \A$, where  $C(\mathbb{S}^1)$ are the continuous functions on the circle $\mathbb{S}^1\cong [0,T)$. We can endow this $C^*$-algebra with the spatial derivation $\nabla_k$, the time derivation $\ii \partial_t$, and the trace given by
$$\widehat{\Tp}(\widehat{A})\;:=\;\int_0^T{\rm d}t\;\Tp\big(A(t)\big),\hspace{1cm} \widehat{A}\in C(\mathbb{S}^1)\otimes \A\;.$$
Thus, it follows that 
$\Delta\mathscr{P}_k\;=\;2\pi\,{\rm Ch}({P}_*)+\mathcal{O}(\varepsilon^N),$ where 
${\rm Ch}({P}_*)$ is the Chern number of the element $\widehat{P}_*\in C(\mathbb{S}^1)\otimes \A$ defined by $t\mapsto P_*(t)$. It is well known that Chern numbers of projections take value in $\Z$ \cite{connes,Shulz2}.
\end{proof}

\medskip

The last result can be rephrased by saying that  up to arbitrarily small corrections in the adiabatic parameter $\varepsilon$, the orbital polarization  $\Delta\mathscr{P}_k$ is topologically quantized.

\section{Applications}\label{apll}
The mathematical framework described in the previous sections applies directly to the two most common cases, namely $\G=\Z^d$ and $\R^d$, which describe discrete (tight-binding) models and continuum systems, respectively. The case of discrete random models has been considered in detail in \cite{Shulz} and it will not be considered here. On the other hand, the treatment of the continuous random case is one of the main motivations for the writing of this work. In the following part of this section, we will present the formalism to describe the continuous random system and we will show that all the 
{\bf Hypothesis 1 - 6} listed in 
 Section \ref{sec:hip} are satisfied for such models.

\subsection*{Continuous models in disordered media}\label{sec: 4,1}
Let us focus  on \emph{ergodic magnetic media} \cite[Section 4]{Bou}, {i.e.},  non-interacting systems of charge particles submitted to a  constant magnetic field $\mathtt{B}$, and to  random potentials $A_\omega$ and $V_\omega$ (solids that can be either random, periodic or quasi-periodic), where $\omega$ runs in the ergodic probability space $(\Omega,\pr)$ of disorder with the ergodic $\R^n$-action $\tau$. Let us consider the one-particle Hilbert space $\mathfrak{h}=L^2(\R^d)$, which describes the quantum states of the system. The constant magnetic field $\mathtt{B}$ can be represented by a $d\times d$ antisymmetric matrix with entries $\{\mathtt{B}_{j,k}\}$
and the associated vector potential $A:\R^d\to\R^d$ can be chosen 
as
$$
A_j(x)\;=\;-\frac{1}{2}\sum_{k=1}^d\mathtt{B}_{j,k}x_k\;,\qquad j=1,\ldots,d\;.
$$
It turns out that
$$
\frac{\partial A_k}{\partial x_j}-\frac{\partial A_j}{\partial x_k}\;=\;\mathtt{B}_{j,k}\;,\qquad \frac{\partial A_k}{\partial x_j}+\frac{\partial A_j}{\partial x_k}\;=\;0\;,\qquad j,k=1,\ldots,d\;.
$$
On  $\mathfrak{h}$ acts the \emph{free Landau Hamiltonian}
$$H_0^A\;:=\;(
-\ii\nabla-A)^2\;,$$
and
 the family of \emph{random} magnetic Hamiltonians
$$H_\omega^A\;\equiv\; H_\omega^A(A_\omega,V_\omega)\;:=\;(
-\ii\nabla-A-A_\omega)^2+V_\omega\;,\qquad\omega\in \Omega\;.$$
In order to ensure  the self-adjointness of the  Hamiltonians $H_0^{A}$ and $H_\omega^{A}$, we assume the \emph{Leinfelder–Simader conditions} on the potentials $A$, $A_\omega$ and $V_\omega$ (see \cite{Lei} or \cite[Section 2.1]{Bou}). It turns out that 
 $H^{A}_\omega$ is essentially self-adjoint on $C_0^\infty(\R^d)$. We will denote with $\mathcal{D}_\omega:=\mathcal{D}(H_\omega^A)$
the domain of $H_\omega^A$, {i.e.}, the closure of 
$C_0^\infty(\R^d)$ with respect to the graph norm induced by $H_\omega^A$. Observe that $H_\omega^A$ meets 
 the \emph{gauge covariance} property
\begin{equation}
  H_\omega^{A+\nabla_\chi}\;=\;\expo{-\ii \chi} H_\omega^{A}\expo{\ii \chi}
\end{equation}
where $\chi:\R^d\rightarrow \R$ is considered as a multiplication operator on $\mathfrak{h}$.

\medskip

Let us consider the direct integral
$$
{\h}\;:=\;\int_\Omega^{\oplus}{\rm d}\;\pr(\omega) \;\mathfrak{h}_\omega\;\simeq\; L^2(\Omega)\otimes L^2(\R^d)\;\simeq\; L^2(\Omega\times\R^d)\;,
$$
and the subspaces $\mathcal{D}:=L^2(\Omega)\otimes\mathcal{D}_\omega$ and $\mathcal{D}_c:=L^2(\Omega)\otimes C_0^\infty(\R^d)$.
The family of Hamiltonians ${H}^{A}:=\{H_\omega^{A}\}_{\omega\in \Omega}$ defines an operator acting on  the Hilbert space ${\h}$. It turns out that  ${H}^{A}$ is essentially self-adjoint with core $\mathcal{D}_c$ and domain $\mathcal{D}$. Moreover, it follows that the maps $\omega\mapsto f(H_\omega^A)$ are measurable for every $f\in L^\infty(\R)$ (see \cite[Section 4.1]{Bou} and references therein). In particular, the spectral projections of  ${H}^{A}$ define measurable maps and this is equivalent to say that 
${H}^{A}$ is affiliated to 
$\Ran({\h})$.

\medskip

Let us introduce the vector-valued   operators $G:=-\ii\nabla+A$. It turns out that every component of $G$ commutes with $H^A_0$.
For every $y$ we consider the unitary operator $T_y:=\expo{-\ii y\cdot G }$ which acts on $\varphi\in \mathfrak{h}$  as
\begin{equation}\label{eq:T}
(T_y\varphi)(x)\;=\;\Theta^\mathtt{B}(y,x)\varphi(x-y)\;=\;\Theta^\mathtt{B}(y,x-y)\varphi(x-y)
\end{equation}
where 
$$
\Theta^\mathtt{B}(y,x)\;:=\expo{\frac{\ii}{2}y\cdot\mathtt{B} \cdot x} \;=\;\expo{\frac{\ii}{2}\sum_{j,k=1}^d\mathtt{B}_{j,k}y_jx_k}\;. 
$$
It follows that $T_yH^A_0T_y^*=H^A_0$ for every $y\in\R^d$. Furthermore, one can check that the map $\Theta^\mathtt{B}:\R^d\times\R^d\to\mathbb{U}(1)$
is a \emph{twisting group $2$-cocycle} according to \cite[Definition 4.1.2]{De Nittis}.
We assume that the potentials $A_\omega$ and $V_\omega$ are covariant random variables, {i.e.}, they meet
 $$
 V_\omega(x-y)\;=\;V_{\tau_y(\omega)}(x)\;,\qquad 
 A_\omega(x-y)\;=\;A_{\tau_y(\omega)}(x)$$ 
for
$\pr$-almost every $\omega\in\Omega$  and Lebesgue-almost every $x\in\R^d$. Then, one obtains the \emph{covariance} relations
$$
T_y H^A_\omega T_y^*=H^A_{\tau_y(\omega)}\;.
$$
If one defines the unitary ${U}_y\in{\B({\h})}$ as
 \begin{equation}\label{eqref_cov}
\big({U}_y\ {\psi}\big)_{\tau_y(\omega)}(x)\;:=\;\Theta^\mathtt{B}(y,x)\;\psi_\omega(x-y)\;,
\end{equation}
 where ${\psi}:=\{\psi_\omega\}_{\omega\in\Omega}$ is any element of ${\h}$ and on the left-hand side the symbol $(\cdot)_{\tau_y(\omega)}$  means the value of the vector ${U}_y{\psi}$ on the fiber of $\h$ at $\tau_y(\omega)$, then one gets the   \emph{invariance} relations
 $$
 {U}_yH^A{U}_y^*\;=\;H^A\;,\qquad \forall y\in\R^d\;.
 $$
 Moreover, one has that the spectral projections of  $H^A$ commute with the $ {U}_y$, and in turn $H^A$ results affiliated with the 
 von Neumann algebra 
$$\A\:=\; {\rm Span}_{\R^d}\{ {U}_y\}'\;\cap\; \Ran( {\h})\;.
$$
Ultimately $H^A\in \Aff(\A)$ according to {\bf Hypothesis 1}.

\medskip

{\bf Hypothesis 2} and ${\bf 3}$ are verified if the $\{ X_1,X_2,...,X_d\}$ are the usual position operators which act constantly  on the fibers of ${\h}$, namely
$(X_j\psi)_\omega(x):=x_j\psi_\omega(x)$ for every $\{\psi_\omega\}_{\omega\in\Omega}\in {\h}$. Observe that the
localization domain
$\mathcal{D}_c:=L^2(\Omega)\otimes C_0^\infty(\R^d)$
is also a core for $H^A$, therefore
$\mathcal{D}_c(H^A):=\mathcal{D}_c\cap \mathcal{D}(H^A)=\mathcal{D}_c$. Finally, the components of the current are $J_k:= \{J_{k,\omega}\}_{\omega\in \Omega}$ with 
$$
J_{k,\omega}\;:=\;2(-\ii\partial_k-A_k-A_{\omega,k})\;,\qquad k=1,\ldots,d\;.
$$

\medskip

The effects of the external deformation on the system are modeled by a sufficiently regular function $w:[0,T]\to\R$ with the boundary  conditions $w(0)=0=w(T)$, which enters in the definition of the 
 time-dependent Hamiltonian ${H}^{A}(t):=\{H_\omega^{A}(t)\}_{\omega\in \Omega}$ defined by
$$H_\omega^A(t)\;:=\;H_\omega^A+w(t)\;W_\omega\;,
$$
where $W:=\{W_\omega\}_{\omega\in \Omega}\in \Ran({\h})$ is a bounded random potential. In view of the Kato-Rellich theorem \cite[Theorem X.12]{reed-simon-II}
one has that $\mathcal{D}(H^A(t))=\mathcal{D}(H^A)$ for every $t\in [0,T]$. Moreover, it is straightforward to check $J_k(t)=J_k$ for every $k=1,\ldots,d$, namely the time-dependent current  equals the stationary current.
Let us assume that there is a \emph{Fermi energy} $\ef\in \C\setminus \sigma(H^A)$ inside the resolvent set of $H^A$. If $\|w\|_\infty\ll 1$ is sufficiently small the gap around $\ef$ doesn't closed during the time-dependent perturbation and one gets that $\ef\in \C\setminus \sigma(H^A(t))$ for every $t\in [0,T]$. This  is in particular a gap condition
stronger than that assumed in {\bf Hypothesis 5}, which is therefore automatically satisfied. In particular the relevant spectral patch
can be chosen as $\sigma_*(t):=(-\infty, \ef]\cap\sigma(H^A(t))$.
In order to complete  the check of the validity of
{\bf Hypothesis 4} we need to prove that  there exists the 
  \emph{unitary propagator}
$[0,T]^2\ni(s,t)\mapsto U^A(s,t)\in \A$ associated to $H^A(t)$. For that, it is sufficient to show that the conditions listed in  
\cite[Theorem X.70]{reed-simon-II} (see also \cite[Section XIV.4]{Y})
are satisfied. The main object is the operator
$$
\begin{aligned}
C(t,s)\;:&=\;\left(H^A(t)-H^A(s)\right) \frac{1}{H^A(s)-\xi{\bf 1}}\\
&=\;
\left(w(t)-w(s)\right) W\frac{1}{H^A(s)-\xi{\bf 1}}\;.
\end{aligned}
$$
If one assumes that $w\in C^1([0,T])$, then $C(t,s)$ automatically fulfills all the conditions for the construction of the unitary propagator.

\medskip

Finally, the relevant initial equilibrium state for $H^A$ is given by the spectral projection of  $H^A$ on the {Fermi energy} $\rho_0:=\chi_{(-\infty,\ef]}(H^A)$. Let us observe that, in view of the gap condition, the 
step function $\chi_{(-\infty,\ef]}$ can be replaced by a smooth and compactly supported function. Therefore, from  \cite[Proposition 4.2]{Bou} one has that also  	{\bf Hypothesis 6} is verified.

  \subsection*{Continuous periodic models}
The case of a continuous periodic  model has been rigorously studied in 
\cite{Pan}. However, it represents a special case of the model described in Section \ref{sec: 4,1} when the {ergodic topological dynamical system} $(\T^d, \R^d,\tau, \mu)$ is given by a $d$-dimensional torus $\T^d:=\R^d/\Gamma$, with $\Gamma\simeq\Z^d$ a lattice, and its normalized Haar measure $\mu$. Evidently, the action of $\R^d$ on $\T^d$ is given by translations and the resulting dynamical system is minimal, which means that the orbit of any point $\omega\in \T^d$ under the action of $\R^d$ is dense. 

\medskip

Let us fix the reference point $\omega_0=[0]$.
In view of the covariance conditions one gets
$$
 V_{\omega_0}(x-\gamma)\;=\;V_{\tau_\gamma(\omega_0)}(x)\;=\; V_{\omega_0}(x)\;,\qquad \forall\; \gamma\in\Gamma
 $$ 
since $\tau_\gamma(\omega_0)=[0+\gamma]=[0]$.
Moreover, this is independent of the election of the  reference point $\omega_0$. Therefore it turns out that  $V_{\omega_0}$, and similarly 
$A_{\omega_0}$ are $\Gamma$-periodic potentials which will be denoted simply by $V_{\Gamma}$ and 
$A_{\Gamma}$, respectively. Note also that for any $\omega\in \mathbb{T}^d$ and $A\in \A$ it holds true that
$$T_y A_{\omega_0} T_y^*\;=\;A_{\tau_y(\omega_0)}\;=\;A_\omega\;,$$
where $\omega=[y].$ Therefore, if one factor the action of $\R^d$ as $\R^d=\T^d\times\Gamma$ one can decompose the algebra $\A$ as follows
$$
\A\;=\;\int_{\T^d}^{\oplus}{\rm d}\;\mu(\omega) \;\mathfrak{A}_\Gamma
$$
where 
$$
\begin{aligned}
\mathfrak{A}_\Gamma\;:&=\;\{A\in\B(L^2(\R^d))\;|\; [T_\gamma,A]=0\;,\;\;\forall\; \gamma\in\Gamma\}\\
&=\;{\rm Span}_{\Gamma}\{ {T}_\gamma\}'
\end{aligned}
$$
is the von Neumann algebra of the bounded operators on the Hilbert space $L^2(\R^d)$ which are invariant under the action of the translations $T_\gamma$ defined by \eqref{eq:T}. Thus, there is a $\ast$-isomorphism of von Neumann algebras $\A\simeq \mathfrak{A}_\Gamma$ given by the identification $\A\ni A\mapsto A_{\omega_0}\in \mathfrak{A}_\Gamma$. Hence, in the case of continuous periodic models, it is sufficient to work with the algebra $\mathfrak{A}_\Gamma$ defined on the Hilbert space $L^2(\R^d)$.

\medskip

In the case that the algebra ${\rm Span}_{\Gamma}\{ {T}_\gamma\}$ contains a commutative $C^*$-subalgebra $\mathfrak{I}_\Gamma$ (rational magnetic flux), then the von Neumann’s spectral Theorem \cite[Part II, Chap.6, Theorem 1]{Dix1}, provides a (new) direct integral decomposition  
\begin{equation}
   L^2(\R^d)\;:=\;\int_{\sigma({\mathfrak{I}_{\Gamma}})}^\oplus {\rm d}\nu(k)\, \h_k
\end{equation}
where $\nu$ is a basic spectral measure and $\sigma(\mathfrak{I}_\Gamma)$ is the Gelfand spectrum of $\mathfrak{I}_\Gamma$. Moreover, there is a unitary map $\mathcal{F}$, called the Bloch-Floquet transform,  such that $\mathcal{F}\mathfrak{A}_\Gamma\mathcal{F}^{-1}$ is contained in the bounded decomposable operators over the direct integral, that is,
$$\mathcal{F}A\mathcal{F}^{-1}\;=\;\int_{\sigma(\mathfrak{I}_\Gamma)}^\oplus {\rm d}\nu(k)\, A(k)\;,\qquad A\in \mathfrak{A}_\Gamma,$$
where $A(k)\in \B(\h_k)$.
Note also that the trace  per unite of volume $\mathscr{T}$ on $\mathfrak{A}_\Gamma$ is given by
\begin{equation}\label{traza}
    \mathscr{T}(A)\;=\;\frac{1}{\mu\big(\sigma(\mathfrak{I}_\Gamma)\big)}\int_{\sigma(\mathfrak{I}_\Gamma)}^\oplus {\rm d}\nu(k)\, {\rm Tr}_{\h_k}\big(A(k)\big)\;.
\end{equation}

\section{Appendix: Adiabatic theorem}\label{append}
The aim of   this section is to extend the \emph{adiabatic Theorem} proved in \cite[Appendix A]{Shulz} to our setting.

\medskip
The first result concerns the regularity of the spectral projections on the gap of $H(t)$.
\begin{lemma}\label{lema 1}
 Suppose that the map $[0,T]\ni t\mapsto H(t)\in \Aff(\A)$ is $N$-differentiable in the uniform sense and that the Hypothesis 1, 2 and 5 hold. Then, the spectral projection map $P_*(t)=\chi_{\sigma_*(t)}(H(t))$ fulfills  $P_*\in C^N([0,T],\A)$.
\end{lemma}
 \begin{proof}
 Let $\gamma(t)\subset \C$ be a closed curve in the resolvent set of $H(t)$ surrounding $\sigma_*(t)$ in the positive sense with 
 $$
 {\rm dist}\big(\gamma (t), \sigma (H(t))\setminus\sigma_*(t)\big)\;\leqslant \;\frac{g}{2}\;,
 $$ 
 where $g$ is defined in  Hypotheses 5. Then
 $$
 P_*(t)\;=\;\frac{1}{{\rm i}2\pi}\oint_{\gamma(t)}{\rm d}z\; \big(z{\bf 1}-H(t)\big)^{-1}\;.
 $$
  Since $f_\pm(t)$ are continuous functions, one has that $\gamma(t+h)$ is homotopic to $\gamma(t)$ in the resolvent set of $H(t+h)$ for $|h|$ small enough, and hence
$$
\begin{aligned}
P_*(t+h)\;&=\;\frac{1}{{\rm i}2\pi}\oint_{\gamma(t+h)}{\rm d}z\; \big(z{\bf 1}-H(t+h)\big)^{-1}\\&=\;\frac{1}{{\rm i}2\pi}\oint_{\gamma(t)}{\rm d}z\; \big(z{\bf 1}-H(t+h)\big)^{-1}\;.
\end{aligned}
$$
Since $\partial_t^n(z{\bf 1}-H(t))^{-1}\in \A$ for all $n\leq N$, then one deduce with an induction on $n$  that
$$
\partial_t^nP_*(t)
\;=\;\frac{1}{{\rm i}2\pi}\oint_{\gamma(t)}{\rm d}z\;  \partial_t^n\big(z{\bf 1}-H(t)\big)^{-1}\;\in\; \A\;.$$
 This concludes the proof.
 %The hypotheses 2 implies that there is a smooth function $f\colon \R\rightarrow\R$ such that $f(H(t))=P_*(t)$. Moreover, the function $f$ can be chosen so that $P_*\in C^N([0,T]), \mathcal{D}(\nabla),$ for all $t\in [0,T],$ which is consequence of the same argument used in [\cite{Bou}, Proposition 4,2].
 \end{proof}

\medskip
The next two results concern the existence of the \emph{superadiabatic projections} and are adaptions of \cite[Proposition 7 and Theorem 9]{Shulz}.

 \begin{proposition}\label{7}
Under the assumptions of the Lemma \ref{lema 1}, there exist  unique maps $P_n\in C^{N+2-n}([0,T], \A)$,  with $1\leqslant n\leqslant N$, such that the functions
 $$
 \widetilde{P}_m^\varepsilon(t)\;=\;\sum_{n=0}^m\varepsilon^nP_n(t)$$
for $0\leqslant m\leqslant N$ and $P_0(t)=P_*(t)=\chi_{\sigma_*(t)}(H(t))$, satisfies
\begin{equation}\label{55}
    \big(\widetilde{P}_m^\varepsilon\big)^2\;=\; \widetilde{P}_m^\varepsilon+\varepsilon^{m+1}G_{m+1}+\mathcal{O}(\varepsilon^{m+2})
\end{equation}
with $G_{m+1}:=\sum_{n=1}^mP_nP_{m+1-n}$ and \begin{equation}\label{56}
    {\rm i}\varepsilon \partial_t\widetilde{P}_m^\varepsilon(t)\;-\;[H(t), \widetilde{P}_m^\varepsilon(t)]_{\dag}\;=\;\ii\varepsilon^{m+1}\partial_t{P}_m(t)\;.
\end{equation}
Furthermore, if $\partial_t^n\big(\ii{\bf 1}-H(t)\big)^{-1}|_{t=t_0}=0$ for some $t_0\in [0,T]$ and all $n\leqslant N,$ then $P_n(t_0)=0$ for all $1\leqslant n\leqslant N.$
 \end{proposition}

 \begin{proof}
This result can be obtained by using induction in $m$. For  $m=0$, with $\widetilde{P}_0^\varepsilon(t):=P_*(t)$ the instantaneous spectral projection of $H(t)$, it follows that
 $$
 \big(\widetilde{P}_0^\varepsilon\big)^2\;=\;\widetilde{P}_0^\varepsilon\;,\qquad{\rm i}\varepsilon\partial_t\widetilde{P}_0^\varepsilon-[H(t),\widetilde{P}_0^\varepsilon]_{\dag}\;=\;\ii\varepsilon \partial_t{P}_\ast(t)\;=\;\mathcal{O}(\varepsilon)\;.$$ 
 Assume now that $\eqref{55}$ and $\eqref{56}$ holds for $P_j$ with $j=0,\ldots,m$. Thus, if we define $P_{m+1}$ as 
\begin{equation*}\label{64}
\begin{aligned}
    P_{m+1}\;:=&\;P_*^\bot G_{m+1}P_*^\bot-P_* G_{m+1}P_*\\
    &+ \frac{1}{{\rm i}2\pi }\oint_{\gamma(t)} {\rm d}z\;(z{\bf 1}-H)^{-1}[\partial_t{P}_m,P_*]_{\dag}(z{\bf 1}-H)^{-1}\;,
\end{aligned}
\end{equation*}
with  $\gamma(t)$ any curve encircling $\sigma_*(t)$ once in the positive sense,
 one can show that $\eqref{55}$ and $\eqref{56}$ hold for $P_{m+1}$
 just following the same  steps in  \cite[Proposition 7]{Shulz}.
  Moreover, since $\A$ is closed under holomorphic functional calculus, one gets $P_{m+1}\in \A$. Finally, if $\partial^n_t\big(\ii{\bf 1}-H(t)\big)^{-1}|_{t=t_0}=0$ then it is also true that $\partial^n_t\big(z{\bf 1}-H(t)\big)^{-1}|_{t=t_0}=0$ for each $z$ in the resolvent of $H(t)$ for every $t$ in $[0,T]$ (see Remark \ref{Rem 1}). Thus, $\dot{P}_*(t_0)=0$ and by the construction of $P_{m+1}$, it follows also that $P_1(t_0)=0$. Using induction again we conclude the last statement.
\end{proof}

\medskip
In order to simplify the notation, we introduce the following norm $$\| A(t)\|_{\mathcal{S},k}\;:=\;\|A(t)\|+\|\partial_tA(t)\|+\|\nabla_k A(t)\|\qquad k=1,\dots,d
$$
for any differentiable path $t\mapsto A(t)$ in $C^{1}([0,T], \A)$.

\begin{theorem}\label{Teo 3}
Let the map $[0,T]\ni t\mapsto H(t)\in \Aff(\A)$ be $N$-differentiable in the uniform sense for some $N\in \N$ and assume the hypothesis  of Lemma \ref{lema 1}. Then, there are constants $\varepsilon_N>0$, $c_N<\infty$ and orthogonal projections $P_N^\varepsilon(t)\in \mathfrak{M}^{1,1}(\A)$ such that the map $(0,\varepsilon_N)\ni\varepsilon \rightarrow P_N^\varepsilon(\cdot)\in C^2([0,T],\mathfrak{M}^{1,1}(\A))$, and the following properties hold uniformly in $t$: 
\begin{equation}\label{53}
    \|P_N^\varepsilon(t)-P_*(t)\|_{\mathcal{S},k}\;<\;c_N\varepsilon\;,
\end{equation}
\begin{equation}\label{54}
   \big \| {\rm i}\varepsilon \partial_tP_N^\varepsilon-\big[\,H(t), P_N^\varepsilon\big]_{\dag}\big\|_{\mathcal{S},k}\;<\;c_N\varepsilon^{N+1}\;.
\end{equation}
Moreover, if $\partial^n_t\big(\ii{\bf 1}-H(t)\big)^{-1}|_{t=t_0}=0$ for some $t_0\in [0, T]$, then $P_N^\varepsilon(t_0)=P(t_0).$
\end{theorem}
\begin{proof}We know by (\ref{55}) that there is a constant $c_N$ such that 
$$\| (\PE)^2-\PE\|
\;\leqslant\; c_N\varepsilon^{N+1}\;.
$$ 
Therefore, the spectral mapping theorem provides
$$
\begin{aligned}
\sigma(\PE)\;&\subset\; [-c_N\varepsilon^{N+1},c_N\varepsilon^{N+1}]\cup[1-c_N\varepsilon^{N+1}, 1+c_N\varepsilon^{N+1}]\\
&\subset \left[-\frac{1}{4},\frac{1}{4}\right]\;\cup\; \left[\frac{3}{4},\frac{5}{4}\right]
\end{aligned}
$$
where the latter holds for $\varepsilon<\varepsilon_N=(4c_N)^{-\frac{1}{N+1}}$. Thus, one can define for any $\varepsilon<\varepsilon_N$
$$
P_N^\varepsilon\;:=\;\frac{1}{{\rm i} 2\pi}\oint_{|z-1|=\frac{1}{2}}{\rm d}z\;(z{\bf 1}-\PE)^{-1},$$
where the integral is taken in the positive sense. It follows that $P_N^\varepsilon\in \A$. Moreover by adapting  the arguments used in \cite[Proposition 4.2]{Bou} one can show that $P_N^\varepsilon(t)\in\mathfrak{M}^{1,1}(\A)$ for every $t\in[0,T]$. Using   the fact that $\PE$ is differentiable one obtains that $\varepsilon \mapsto P_N^\varepsilon(\cdot)$ is in $ C^2([0,T], \mathfrak{M}^{1,1}(\A))$. Now one can obtain \eqref{53} by  taking the norms of
$$
P_N^\varepsilon-P_*\;=\; \frac{1}{{\rm 1} 2\pi}\oint_{|z-1|=\frac{1}{2}}{\rm d}z\;(z{\bf 1}-\PE)^{-1}(\PE-P_*)(z{\bf 1}-P_*)^{-1}\;,$$
of its time derivate $\partial_t$ and of its gradient $\nabla_k$. In the same way, we can use
\begin{equation*}
    \begin{split}
       &{\rm i}\varepsilon \partial_tP_N^\varepsilon- [H,P_N^\varepsilon]_{\dag}\\
       &=\; \frac{1}{{\rm i} 2\pi}\oint_{|z-1|=\frac{1}{2}}{\rm d}z\;\left({\rm i}\varepsilon \partial_t(z{\bf 1}-\PE)^{-1}-[H, (z{\bf 1}-\PE)^{-1}]_{\dag}\right)\\
        &=\;\frac{1}{{\rm i} 2\pi}\oint_{|z-1|=\frac{1}{2}}{\rm d}z\;(z{\bf 1}-\PE)^{-1}\left({\rm i}\varepsilon \partial_t\PE-[H, \PE]_{\dag}\right)
        (z{\bf 1}-\PE)^{-1}\\
        &=\;\frac{\varepsilon^{N+1}}{{\rm i} 2\pi}\oint_{|z-1|=\frac{1}{2}}{\rm d}z\;(z{\bf 1}-\PE)^{-1}\partial_t{P}_N(z{\bf 1}-\PE)^{-1}
    \end{split}
\end{equation*}
to show \eqref{54}. The last claim follows directly from Proposition \ref{7}.
\end{proof}

\medskip

The proof of the following result is an adaption of \cite[Corollary 5]{Shulz}.
\begin{corollary}\label{Coro 3}
Let $\rho_{sa}^\varepsilon(t)$ be the unique solution of the equation
$$ {\rm i}\varepsilon\partial_t\rho^{\varepsilon}_{sa}(t)\;=\;\ii[H(t),\rho_{sa}^\varepsilon(t)]_{\dag}\;,\qquad \rho_{sa}^\varepsilon(0)\;:=\;P_N^{\varepsilon}(0).$$
Then under the hypothesis  of Lemma \ref{lema 1} one gets that
\begin{equation*}
    \rho_{sa}^\varepsilon(t)\;=\;P_N^\varepsilon(t)+\Delta^\varepsilon(t)
\end{equation*}
with $\|\Delta^\varepsilon(t)\|_\mathcal{S}=\mathcal{O}(\varepsilon^N|t|)$.
\end{corollary}

\section*{Acknowledgements}
GD's research is supported by the grant {Fondecyt Regular - 1190204}. DP’s research is supported by ANID-Subdirección    de Capital Humano/ Doctorado Nacional/ 2022-21220144.
GD's would like to thank the Alexander von Humboldt Foundation for supporting his stay at the University of  Erlangen-N\"urnberg during July 2022 where the large part of this work was completed. He is also grateful to Camping due barche (Scanzano, Italy) where the peaceful atmosphere of this place provided an invaluable help for the preparation of the final version of this manuscript. The authors are  indebted to M. Lein and S. Teufel for many stimulating discussions.

\input{references2.tex}

\end{document}

%% file: references2.tex
%%%%%%%%%%%%%%%%%%%%%%%% referenc.tex %%%%%%%%%%%%%%%%%%%%%%%%%%%%%%
% sample references
% %
% Use this file as a template for your own input.
%
%%%%%%%%%%%%%%%%%%%%%%%% Springer-Verlag %%%%%%%%%%%%%%%%%%%%%%%%%%
%
% BibTeX users please use
% \bibliographystyle{}
% \bibliography{}
%
%\biblstarthook{References may be \textit{cited} in the text either by number (preferred) or by author/year.\footnote{Make sure that all references from the list are cited in the text. Those not cited should be moved to a separate \textit{Further Reading} section or chapter.} If the citatiion in the text is numbered, the reference list should be arranged in ascending order. If the citation in the text is author/year, the reference list should be \textit{sorted} alphabetically and if there are several works by the same author, the following order should be used:
%\begin{enumerate}\item all works by the author alone, ordered chronologically by year of publication\item all works by the author with a coauthor, ordered alphabetically by coauthor\item all works by the author with several coauthors, ordered chronologically by year of publication.\end{enumerate}}